\documentclass[a4paper]{article}

\usepackage{times}
\usepackage{mathtools}
\usepackage{amsthm}
\usepackage{amssymb}
\usepackage{cite}
\usepackage{tikz}
\usepackage{paralist}
\usepackage{fixme,manfnt,mparhack}
\usepackage[basic]{complexity}
\usepackage{authblk}
\usepackage[T1]{fontenc}
\usepackage{enumerate}
\usepackage{algorithm}
\usepackage{algorithmic}

\usetikzlibrary{positioning,fit,arrows,automata,calc,shapes} 
\tikzset{->,>=stealth',shorten >=1pt,shorten <=1pt,auto,node distance=2cm,
every loop/.style={looseness=6},
initial text={},
every fit/.style={draw,densely dotted,rectangle},
el/.style={font=\footnotesize},
va/.style={font=\scriptsize}}

\newcounter{fixcount}
\newcommand{\defineNote}[3][black!65!green]{\expandafter\def\csname #2\endcsname
##1{\stepcounter{fixcount}\fxwarning{\textcolor{#1}{\textbf{#3}: ##1}}}}
\defineNote{PH}{Paul}
\defineNote{JFR}{Jean-Fran\c cois}

\newcommand{\val}[3]{{\tt v}^{#1}_{#2#3}}
\newcommand{\maxmax}{\max\!\max}
\newcommand{\maxmin}{\max\!\min}
\newcommand{\minmax}{\min\!\max}
\newcommand{\minmin}{\min\!\min}

\newcommand{\PG}{\textsc{Parity game}}
\newcommand{\xMP}{\textbf{MP}}

\newcommand{\st}{\mathrel{:}}

\newtheorem{lemma}{Lemma}
\newtheorem{corollary}{Corollary}
\newtheorem{theorem}{Theorem}
\newtheorem{conjecture}{Conjecture}

\title{Quantitative games with interval objectives}
\author{Paul Hunter \and Jean-Fran\c{c}ois Raskin}
\date{April 2014}

\begin{document}
\maketitle

\begin{abstract}
Traditionally quantitative games such as mean-payoff games and
discount sum games have two players -- one trying to maximize the
payoff, the other trying to minimize it.
The associated decision problem, ``Can Eve (the maximizer) achieve, for example, a
positive payoff?'' can be thought of as one player trying to attain a
payoff in the interval $(0,\infty)$.
In this paper we consider the more general problem of determining if a
player can attain a payoff in a finite union of arbitrary intervals for various payoff
functions (liminf, mean-payoff, discount sum, total sum).  In
particular this includes the interesting exact-value problem, ``Can Eve
achieve a payoff of exactly (e.g.) 0?''
\end{abstract}

\section{Introduction}

Quantitative two-player games on graphs have been extensively studied in the verification community~\cite{em79,cdhr10,Rei13,gs09,zp96}. Those models target applications in reactive system synthesis with resource constraints.  In these games two players, Eve and Adam, interact by moving a token around a weighted, directed graph, for a possibly infinite number of moves. This interaction results in a play which is an infinite path in the graph.  The value of the play is computed by applying a payoff function to the sequence of weights of the edges traversed along the path.  
Typical payoff functions are (lim)sup, (lim)inf, mean-payoff, (total) sum, and discounted sum. 

In the literature is usual to assume that Eve is attempting to maximize the payoff and Adam is attempting to minimize it.  In this context all these games are determined, that is the maximum that Eve can ensure is equal to the minimum that Adam can ensure, and this value can be computed in polynomial time for (lim)inf and (lim)sup~\cite{cdh13}, and in pseudo-polynomial time for mean-payoff, discounted sum, and total sum~\cite{zp96,gs09}.  The associated decision problem is the \emph{threshold problem}: Given a game graph, a payoff function and a threshold $\nu$ does Eve have a strategy to ensure all consistent plays have payoff at least $\nu$?   The threshold problems  for the aforementioned payoff functions are all closely related, and it is known that Eve and Adam can play optimally in those games with {\em memoryless strategies}~\cite{gz04}.  Consequently the decision problem for all those games is in $\NP \cap \coNP$. In fact, it can be shown in $\UP \cap \coUP$ for mean-payoff, discounted sum, and total sum, and in $\P\TIME$ for (lim)inf and (lim)sup.

The threshold problem can be seen as game in which Eve is trying to force the payoff to belong to the interval of values $[\nu,\infty)$.  In this paper we consider the more general problem of determining if a player can attain a payoff in a finite union of arbitrary intervals for the classical payoff functions mentioned above. That is, we are interested in the following question: Given a weighted arena $G$ and a finite union of real intervals, what is the complexity 
of determining if Eve has a winning strategy to ensure the payoff of any consistent play lies within the interval union?  In particular this includes the interesting exact-value problem: Can Eve achieve a payoff of exactly $\nu$?  Such objectives arise when considering efficiency constraints, for example can a system achieve a certain payoff without exceeding a certain target?  
We consider two versions of our problem depending on whether the numeric inputs (weights, interval bounds and discount factor) are given in binary or unary.  We also consider the memory requirements for a winning strategy both for Eve and Adam.  Our games are a natural subclass of multi-dimensional quantitative games (see e.g.~\cite{cdhr10}), however our results are largely incomparable with that paper as we consider a wider array of payoff functions and  our objective corresponds to \emph{disjunctions} of multi-dimensional objectives which were not considered.

\begin{table}[h]
\begin{center}
\small
\begin{tabular}{|l|c|c|c|}
\hline
Payoff type & Single interval & \multicolumn{2}{c|}{Multiple intervals}\\
\cline{3-4}
&&Binary&Unary\\
\hline\hline
Liminf/limsup & \P\TIME & $\NP\cap \coNP$& \PG-c\\
\hline
Mean-payoff & $\NP \cap \coNP$ & \PSPACE & \PG-hard\\
\hline
Discounted sum (non-singleton) & \multicolumn{2}{c|}{\PSPACE-c}&\P\TIME\\
\hline
Discounted sum (exact value) & \multicolumn{2}{c|}{\PSPACE-hard}&?\\
\hline
Total sum & \multicolumn{2}{c|}{\EXP-hard, \EXPSPACE}&\PSPACE-c\\
\hline
\end{tabular}
\caption{Complexity of deciding the winner in interval games}\label{tab:complexitySummary}
\end{center}
\end{table}

\begin{table}
\begin{center}
\small
\begin{tabular}{|l|c|c|}
\hline
Payoff type & Single interval & Multiple intervals\\
&(Eve/Adam)&\\
\hline\hline
Liminf/limsup & \multicolumn{2}{c|}{Positional}\\
\hline
Mean-payoff & Finite/Positional& Infinite\\
\hline
Discounted sum (non-singleton) & \multicolumn{2}{c|}{Finite} \\
\hline
Discounted sum (exact value) & \multicolumn{2}{c|}{Infinite} \\
\hline
Total sum & Finite/Infinite & Infinite\\
\hline
\end{tabular}
\caption{Memory requirements for interval games}\label{tab:memorySummary}
\end{center}
\end{table}

Tables~\ref{tab:complexitySummary} and~\ref{tab:memorySummary} summarize the results of this paper: the first table highlights the complexity results and the second table highlights the memory requirements for playing optimally.
While the classical threshold problems for weighted games can be solved in $\P\TIME$ for (lim)inf and (lim)sup and in $\NP \cap \coNP$ for mean-payoff, discounted sum and total sum, and memoryless strategies always suffice, the situation for our interval objectives is far richer:
\begin{itemize} 
\item For liminf and limsup, we provide a polynomial time algorithm in the case of a single interval. For a union of intervals, we show that these games are polynomially equivalent to parity games: so we can solve them in $\NP \cap \coNP$, and a polynomial time algorithm for interval liminf games would provide a polynomial time algorithm for parity games (a long-standing open question in the area). Optimal strategies are memoryless for both players.
\item For interval mean-payoff games, we provide a recursive algorithm that executes in polynomial space. This algorithm leads to a $\NP \cap \coNP$ algorithm in the case of single interval objectives. While mean-payoff games can be solved in polynomial time when weights are given in unary, we show here that interval mean-payoff games are at least as hard as parity games even when weights are given in unary. So, a pseudo-polynomial time algorithm for interval mean-payoff games would lead to a polynomial algorithm for parity games. For a union of intervals, infinite memory may be necessary for both players, and for single interval exponential memory may be  necessary for Eve while Adam can always play a memoryless strategy. 
\item Interval discounted sum games are complete for polynomial space when singleton intervals (and singleton gaps between intervals) are forbidden. The decidability for the case when singletons are allowed is left open and it generalizes known open problems in single player discounted sum graphs~\cite{cfw13,bh11}. Finite memory suffices for both players in the non-singleton case and infinite memory is needed for both players when singletons are allowed.
\item For the total sum payoff, we establish a strong link with one counter parity games that leads to a \PSPACE-complete result for unary encoding and an \EXPSPACE\ solution for the binary encoding together with an \EXP-hardness result. For single interval games Eve need only play finite memory strategies, while she may need infinite memory in the general case.  In both cases, Adam may require an infinite memory strategy.

\end{itemize}

\paragraph{{\bf Structure of the paper}} Section~\ref{sec:prelim} introduces the necessary preliminaries. In Sections \ref{sec:liminf}, \ref{sec:MP}, \ref{sec:DS}, and~\ref{sec:totalSum} we consider the decision problems and memory requirements for the liminf/limsup, mean-payoff, discounted sum, and total sum payoff functions, respectively.

\section{Preliminaries}\label{sec:prelim}
A game graph is a tuple $G = (V,V_\exists,E,w,q_0)$ where $(V,E,w)$ is an edge-weighted graph, $V_\exists \subseteq V$, and $q_0 \in V$ is the initial state.  Without loss of generality we will assume all weights are integers. 
In the sequel we will depict vertices in $V_\exists$ with squares and vertices in $V \setminus V_\exists$ with circles.  In complexity analyses we will denote the maximum absolute value of a weight in a game graph by $W$.  If $V' \subseteq V$, we denote by $G\setminus V'$ the game graph induced by $V\setminus V'$.

A play in a game graph is an infinite sequence of states $\pi = v_0v_1 \cdots$ where $v_0 = q_0$ and $(v_i,v_{i+1}) \in E$ for all $i$.  Given a play $\pi = v_0v_1 \cdots$ and integers $k,l$ we define $\pi[k..l] = v_k\cdots v_l$, $\pi[..k] = \pi[0..k]$, and $\pi[l..] = v_lv_{l+1}\cdots$.  We extend the weight function to partial plays by setting $w(\pi[k..l]) = \sum_{i=k}^{l-1} w((v_i,v_{i+1}))$. A strategy for Eve (Adam) is a function $\sigma$ that maps partial plays ending with a vertex $v$ in $V_\exists$ ($V \setminus V_\exists$) to a successor of $v$.  A strategy has memory $M$ if it can be realized as the output of a finite state machine with $M$ states.  A memoryless (or positional) strategy is a strategy with memory $1$, that is, a function that only depends on the last element of the given partial play.  A play $\pi = v_0v_1 \cdots$ is consistent with a strategy $\sigma$ for Eve (Adam) if whenever $v_i \in V_\exists$ ($v_i \in V\setminus V_\exists$), $\sigma(\pi[..i]) = v_{i+1}$.

\subsection{Payoff functions}
A play in a game graph defines an infinite sequence of weights.  We define below several common functions that map such sequences to real numbers.

\paragraph{Liminf/limsup.}
The liminf (limsup) payoff is determined by the minimum (maximum) weight seen infinitely often.  Given a play $\pi = v_0v_1\cdots$ we define:
\[ \liminf(\pi) = \liminf_{i \to \infty} w(v_i,v_{i+1}) \qquad \limsup(\pi) = \limsup_{i \to \infty} w(v_i,v_{i+1}).\]
Note that by negating all weights and the endpoints of the intervals we transform a limsup game to a liminf game and vice-versa.  

\paragraph{Mean-payoff.}
The \emph{mean-payoff} value of a play is the limiting average weight, however there are several suitable definitions because the running averages might not converge.  The mean-payoff values of a play $\pi$ we are interested in are defined as:
\[ \underline{MP}(\pi) = \liminf_{k \to \infty} \frac{1}{k} w(\pi[..k]) \qquad \overline{MP}(\pi) = \limsup_{k \to \infty} \frac{1}{k} w(\pi[..k]).\]
As with liminf/limsup games we can switch between definitions by negating weights and interval endpoints, so we will only consider the $\underline{MP}$ function.

\paragraph{Discounted sum.}
The \emph{discounted sum} is defined by a discount factor $\lambda \in (0,1)$.  Given a play $\pi=v_0v_1\cdots$, we define:
\[ DS_\lambda(\pi) = \sum_{i=0}^\infty \lambda^i \cdot w(v_i,v_{i+1}).\]

\paragraph{Total sum.}
The \emph{total sum} condition can be thought of as a refinement of the mean-payoff condition, enabling discrimination between plays that have a mean-payoff of $0$.  Given a play $\pi$ we define:
\[ \underline{Total}(\pi) = \liminf_{k \to \infty} w(\pi[..k]) \qquad \overline{Total}(\pi) = \limsup_{k \to \infty} w(\pi[..k]).\]
As with liminf/limsup games we can switch between definitions by negating weights and interval endpoints, so we will only consider the $\underline{Total}$ function.

\subsection{Interval games}
For a fixed payoff function $F$, an  \emph{interval $F$ game} consists of a finite game graph and a finite union of real intervals $I = I_1 \cup \cdots \cup I_r$.  Given an interval $F$ game $(G,I)$, a play $\pi$ in $G$ is winning for Eve if $F(\pi) \in I$ and winning for Adam if $F(\pi) \notin I$.  We say a player wins the interval game if he or she has a strategy $\sigma$ such that all plays consistent with $\sigma$ are winning for that player.
For convenience we will assume the intervals are non-overlapping and ordered such that $\sup I_i \leq \inf I_{i+1}$ for all $i$.

\subsection{Parity games}
A parity game is a pair $(G,\Omega)$ where $G$ is a game graph (with no weight function) and $\Omega:V \to \mathbb{N}$ is a function that assigns a priority to each vertex.  Plays and strategies are defined as with interval games.  A play defines an infinite sequence of priorities, and we say it is winning for Eve if and only if the minimal priority seen infinitely often is even.

\section{Liminf games}\label{sec:liminf}
The first payoff function we consider is the $\liminf$ function.  Note that as this always takes integer values, we can assume all intervals are closed or open as necessary.
We show below that deciding interval liminf games is polynomially equivalent to deciding parity games.  In particular the number of intervals is equal to the number of even priorities required, so single interval liminf games are equivalent to parity games with at most three priorities and can therefore be solved in polynomial time~\cite{Sch07}.   Further, the range of the priorities are determined by range of the weight function and vice versa, so this equivalence also holds for unary encoded interval liminf games.
\begin{theorem}
The following problems are polynomially equivalent:
\begin{enumerate}[(i) ]
\item Deciding if Eve wins a unary encoded interval liminf  game;
\item Deciding if Eve wins a binary encoded interval liminf  game; and
\item Deciding if Eve wins a parity game.
\end{enumerate}
\end{theorem}

\begin{proof}
(i)$\Rightarrow$(ii): Trivial.

(ii)$\Rightarrow$(iii): 
For this reduction, we use the following function which will also be used in Section~\ref{sec:totalSum}.  Let $I = I_1 \cup I_2 \cup \cdots \cup I_r$ be a finite union of closed integer intervals such that $\sup I_i < \inf I_{i+1}$ for all $i$.  Define $\Omega_I:\mathbb{Z} \to [1,2r+1]$ as follows:
\[ \Omega_I(n) = \left\{\begin{array}{ll} 2i&\text{ if $n \in I_i$},\\
1&\text{ if $n < \inf I_1$, and}\\
\max\{1+2i\st \sup I_i < n\}&\text{ otherwise.}\end{array}\right.
\]

Now suppose $(G,I)$ is an interval liminf game.  
We transform the game graph $G$ to $G'$ as follows.  Every edge $e$ is sub-divided and the subdividing vertex is given priority $\Omega_I(w(e))$.
The original vertices of $G$ are all given priority $2r+1$.

It is not difficult to see that there is a $1$-$1$ correspondence between plays in $G$ and plays in $G'$, and that for any play in $G$, $\liminf w(e) \in I_i$ for some $i$ if and only if the minimum priority in the corresponding play in $G'$ seen infinitely often is even.

(iii)$\Rightarrow $(i): 
To go the other direction, given a parity game played on $G$ we transform it to an interval liminf game played on $G'$ as follows.  $G'$ is the weighted graph obtained by setting the weight of an edge to be the priority at the vertex at the tail of the edge (that is, the vertex for which the edge is outgoing).  The intervals are singleton intervals containing each of the even priorities that occur in $G$.  Clearly any play in $G$ is a play in $G'$ and it is not difficult to see that for a play in $G$ the minimum priority seen infinitely often is even if and only if the $\liminf$ of the weights of all edges in a play of $G'$ lie in a given interval.
\end{proof}

We observe that the above reductions between parity and liminf games do not significantly alter the topology of the game graph (if at all).  In particular, positional strategies in one game readily translate to positional strategies in the other.  It follows from the positional determinacy of parity games~\cite{Zie98}, that:
\begin{corollary}
Positional strategies suffice for interval liminf games.
\end{corollary}
\section{Mean-payoff games}\label{sec:MP}



In this section we investigate interval mean-payoff games.  We give a recursive algorithm that repeatedly asks for a solution for the \emph{mean-payoff threshold problem}: Given a game graph $G$ and a threshold $\nu \in \mathbb{Q}$ does Eve have a strategy to ensure the (liminf) mean-payoff of all consistent plays is at least\footnote{or at most if she is minimizing the payoff} $\nu$?  As mentioned earlier this problem is known to be in $\NP \cap \coNP$, and solvable in time $O(|V|\cdot|E|\cdot W)$ and space $O(|V| \cdot \log(|E|\cdot W))$~\cite{bcdgr11}. We denote this problem by $\xMP_{\sim \nu}(G)$ where $\sim \in \{\geq, >, \leq, <\}$ depending on whether Eve is maximizing or minimizing the payoff and whether or not a payoff of $\nu$ is winning for Eve.  It is well known~\cite{em79} that the strict threshold problem can be reduced to a non-strict threshold problem -- this follows from the fact that mean-payoff values are restricted to a finite set of rationals.

Our algorithm implies that for a fixed number of intervals the problem reduces to the classic threshold problem (under polynomial-time Turing reductions).  In Section~\ref{sec:singleMP} we consider single interval mean-payoff games in more detail.  In particular we show that in this case finite memory strategies (indeed, positional strategies for Adam) suffice for winning strategies.  However, our first observation of this section is that in general interval mean-payoff games may require infinite memory.
\begin{lemma}
Finite memory winning strategies are not sufficient in interval mean-payoff games.
\end{lemma}
\begin{proof}
Consider the game in Figure~\ref{fig:infMem} where $I = (0,1] \cup [2,\infty)$.  Eve has an infinite memory winning strategy in this game as follows.  First she plays to $q_1$.  Then she counts how many times Adam takes the loop $(q_1,q_1)$.  If Adam returns to $q_0$ then Eve takes the loop $(q_0,q_0)$ the same number of times before returning to $q_1$.  Clearly any play consistent with this strategy that only visits $q_0$ finitely often will satisfy $\underline{MP} = 2$, and any play that visits $q_0$ infinitely often will satisfy $\underline{MP}=1$.  Therefore the strategy is winning for Eve.  Now suppose Eve plays a finite memory strategy $\sigma$ with memory $M$.  We observe that any play consistent with $\sigma$ that visits $q_0$ either remains in $q_0$ or exits $q_0$ in at most $M$ steps -- if a play stays in $q_0$ for more than $M$ steps then a memory state must have been revisited, thus the strategy will keep the play in $q_0$ indefinitely.  Consider the following (finite memory) strategy of Adam: whenever the play reaches $q_1$, take the loop $(q_1,q_1)$ $M+1$ times then move to $q_0$.  We claim this strategy is winning for Adam.  If at some point the play consistent with $\sigma$ and this strategy remains in $q_0$ indefinitely then it has $\underline{MP}=0$, so it is winning for Adam.  Otherwise the play exits $q_0$ infinitely often, that is the edge $(q_0,q_1)$ is taken infinitely often.  Let us break up the play into the segments defined by successive occurrences of this edge.  Following the above argument the length of each of these segments is between $M+3$ and $2M+3$, and the weight of each of these segments is exactly $2M+4$.  Thus the average weight for each segment lies between $1+\frac{1}{2M+3}$ and $2-\frac{2}{M+3}$ inclusive.  As $M$ is fixed, it follows that $\underline{MP} \in (1,2)$ and thus the play is winning for Adam.
\end{proof}

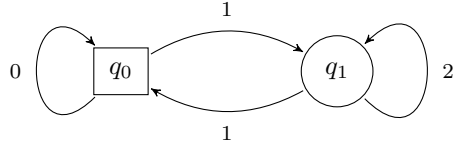
\begin{figure}
\begin{center}
\begin{tikzpicture}[inner sep=2mm, ve/.style={rectangle, draw}, va/.style={circle, draw}]
\node[ve] (A){$q_0$};
\node[va, right=of A] (B){$q_1$};
\path
(A) edge[bend left] node[el]{$1$} (B)
(B) edge[loop, looseness=6,in=45,out=-45] node[swap,el]{$2$} (B)
(B) edge[bend left] node[el]{$1$} (A)
(A) edge[loop, looseness=6,in=135,out=225] node[el]{$0$} (A);

\end{tikzpicture}
\end{center}
\caption{Interval mean-payoff game ($I = (0,1] \cup [2,\infty)$) which requires infinite memory}\label{fig:infMem}
\end{figure}


\subsection{Upper bounds}

We now present an algorithm, Algorithm~\ref{alg:MPfull}, for computing the winning regions in an interval mean-payoff game.

\begin{algorithm}
\caption{$\xMP_{I}(G)$}\label{alg:MPfull}
\begin{algorithmic}
\REQUIRE A game graph $G = (V,V_\exists,E,w,q_0)$ and a finite union of real intervals $I$.
\ENSURE $(W^\exists,W^\forall)$ where $W^\exists$ ($W^\forall$) are the vertices from which Eve (Adam) has a winning strategy.
\IF{$I = \emptyset$}
\RETURN $(\emptyset, V)$
\ENDIF
\STATE $a \leftarrow \inf I$
\IF{$a = -\infty$}
\STATE $(W,W') \leftarrow \xMP_{\mathbb{R} \setminus I}(\overline{G})$  \COMMENT{$\overline{G}$ is $G$ with $V_\exists$ and $V\setminus V_\exists$ swapped}
\ELSE
\STATE $W \leftarrow \emptyset$
\REPEAT
\STATE $(A,A') \leftarrow \xMP_{\succ a}(G)$ \COMMENT{If $a \in I$ then $\succ = \geq$ otherwise $\succ = >$}
\STATE $(B,B') \leftarrow \xMP_{(-\infty, a] \cup I}(G)$
\STATE $W \leftarrow W \cup A' \cup B'$
\STATE $G \leftarrow G \setminus (A' \cup B')$
\UNTIL{$A' \cup B' = \emptyset$}
\ENDIF
\RETURN $(V \setminus W, W)$

\end{algorithmic}
\end{algorithm}

The correctness of the algorithm is given by the following lemma.
\begin{lemma}\label{lem:MPCorrect}
Let $(G,I)$ be an interval mean-payoff game.  $\xMP_I(G)$ correctly computes the winning regions for Adam and Eve.
\end{lemma}
\begin{proof}
We observe that by symmetry the winning regions of $\xMP_I(G)$ are precisely the complements of the winning regions of $\xMP_{\mathbb{R}\setminus I}(G)$.  Thus the algorithm correctly computes the winning regions for $I$ if and only if correctly computes the winning regions for $\mathbb{R}\setminus I$.  In particular we can assume that either $I=\emptyset$ or $\inf I > -\infty$.  

The proof is by induction on the number of interval boundaries in $I$.  If there are no boundaries then $I=\emptyset$ and so $\xMP_I(G)$ returns the correct value: $(\emptyset, V)$.  Now suppose $a = \inf I > -\infty$.  Note that $I' = (-\infty,a] \cup I$ has one interval boundary fewer than $I$, so by the induction hypothesis the recursive call in line 11 correctly computes the winning regions of $G$ for the interval $I'$.  Let $W_i$ ($i=0,1,\ldots$) denote the set of vertices in $W$ after $i$ iterations.  Note that the algorithm runs until $W_n = W_{n+1}$, and the subgraph of $G$ used in the $i$-th iteration is $G \setminus W_{i-1}$.  We prove by induction on $i$ that Adam has a winning strategy from every vertex in $W_i$.  For $i=0$, $W_0 = \emptyset$ so the result holds trivially.  Now suppose Adam has a winning strategy from every vertex in $W_i$, and let $v \in W_{i+1} \setminus W_i$.  Either $v$ is in the winning region of Adam for $\xMP_{\succ a}(G\setminus W_i)$ or $v$ is in the winning region of Adam for $\xMP_{I'}(G\setminus W_i)$.  In both cases the corresponding winning strategy will ensure a payoff outside $I$ and will therefore be winning for plays restricted to $G \setminus W_i$.  Thus his strategy from $v$ is to play this strategy until a vertex in $W_i$ is reached, whereupon he switches to the winning strategy from that vertex.

We now show that Eve has a winning strategy on the vertices in $V \setminus W$.  Note that on these vertices Eve has two strategies: a memoryless strategy $\sigma_>$ which ensures $\underline{MP} \succ a$; and, by the inductive hypothesis, a strategy $\sigma_<$ which ensures a payoff in the interval $I'$.  Also note that plays consistent with these strategies remain in $V \setminus W$.  We now show how to combine these two strategies to obtain a winning strategy for the interval $I$.  For simplicity we will assume $a\in I$, if it is not the case, then the same arguments apply by replacing $a$ with the smallest payoff Adam can attain against $\sigma_>$.
Let $I_1$ be the interval of $I$ with $a = \inf I_1$, and let $t$ be any element of $I_1$.  The strategy for Eve is to track the current average weight of the play so far.  If it is less than $t$ then she plays $\sigma_>$ and if it is greater than or equal to $t$ then she plays $\sigma_<$.  Clearly if she changes strategy only finitely often then her strategy is winning: if she eventually only plays $\sigma_>$ then the payoff will be in $[a,t) \subseteq I_1 \subseteq I$; and if she eventually only plays $\sigma_<$ then the payoff will be in $[t,\infty) \cap I' \subseteq I$.  Now suppose the play causes Eve to switch strategy infinitely often.  The problem here is that when switching to $\sigma_>$ the average weight may go below $a$, and if this happens infinitely often the $\liminf$ average may be below $a$.  However, as $\sigma_>$ is memoryless, the average after $n$ steps will never be more than $\frac{(|V|+1)W}{n}$ below $a$: this is seen easiest by taking $a=0$ and considering the total, rather than the average, weight.  This tends to $0$ as $n$ tends to $\infty$ hence $\underline{MP}$ is at least $a$.  As the average goes below $t$ infinitely often, $\underline{MP} \leq t$.  Therefore the payoff of the play is in $[a,t] \subseteq I_1 \subseteq I$, and hence the combined strategy is winning for Eve.
\end{proof}

The running time for Algorithm~\ref{alg:MPfull} is $|V|^{2r-1} \cdot {\bf MP}$, where ${\bf MP}$ is the running time for an algorithm to solve the mean-payoff threshold problem.  It is straightforward to see that the algorithm can be implemented in polynomial space. 

\begin{theorem}
Let $G$ be a game graph and $I$ a finite union of $r$ real intervals.
Whether Eve wins the interval mean-payoff game $(G,I)$ can be decided in time $O(|V|^{2r}\cdot |E| \cdot W)$ and space $O(r \cdot |V|\cdot \log (|E|\cdot W))$.
\end{theorem}

We observe that although the players may require infinite memory for a winning strategy, Algorithm~\ref{alg:MPfull} shows that a winning strategy can be succinctly represented by $2r$ positional sub-strategies.  It is not clear
that given such a certificate whether there exists an efficient algorithm for computing the winning region, however we believe that this is the case.  By the symmetry of the roles of the players, such an algorithm would show that
the interval mean-payoff game is both in $\NP$ and $\coNP$.
\begin{conjecture}
Determining whether Eve wins an interval mean-payoff game is in $\NP \cap \coNP$.
\end{conjecture}

\subsection{Lower bound}
The above conjecture would hold if we could solve interval mean-payoff games with only polynomially many calls to the mean-payoff threshold problem.
We now give a lower bound for the complexity of deciding interval mean-payoff games which suggests any such algorithm would yield quite remarkable results: we reduce parity games to interval mean-payoff games with small weights and small interval bounds.
In particular this implies that any pseudo-polynomial time algorithm (including polynomially many calls to the threshold problem) would yield a polynomial time algorithm for parity games.

\begin{theorem}
There is a polynomial time reduction from parity games to unary-encoded interval mean-payoff games.
\end{theorem}
\begin{proof}
Let $(V,V_\exists, E, q_0, \Omega)$ be a (min-)parity game.  Without loss of generality we can assume that  the set of priorites is contained in $[0,|V|]$.  
We construct an interval mean-payoff game $(V',V'_\exists, E', w, q_0', I)$ as follows.
\begin{itemize}
\item $I = [0,1) \cup [2,3) \cup \cdots \cup [n,n+1)$ where $n$ is the smallest even integer greater than or equal to $|V|$;
\item $V' = V \cup V \times \{0,+,-\}$.  For simplicity we write $(v,\ast)$ as $v^\ast$;
\item $q_0' = q_0$;
\item $V'_\exists = V_\exists \cup \{v^0,v^+,v^- \st \Omega(v)\text{ is even}\}$;
\item $E'$ and $w$ are constructed as follows:
\begin{itemize}
	\item For each $(v,w) \in E$, $(v,w^0) \in E'$ and the weight of this edge is $\Omega(v)$,
	\item For each $v \in V$: $(v^0,v^+), (v^0,v^-), (v^+,v), (v^-,v) \in E'$ all with weight $\Omega(v)$, $(v^+,v^+) \in E'$ with weight $\Omega(v)+1$, and $(v^-,v^-) \in E'$ with weight $\Omega(v)-1$.
\end{itemize}
\end{itemize}
Intuitively, we replace each vertex in the original game with the gadget shown in Figure~\ref{fig:parityGadget}.  If the priority of the vertex is even then the gadget is controlled by Eve, and if it is odd then it is controlled by Adam.  The last vertex in the gadget is controlled by the player that controlled the original vertex.

\begin{figure}
\begin{center}
\begin{tikzpicture}[inner sep=2mm, ve/.style={rectangle, draw}, va/.style={circle, draw}]
\node[ve] (A){$v^0$};
\node[ve, right=of A, yshift=0.5cm] (A+){$v^+$};
\node[ve, right=of A, yshift=-0.5cm] (A-){$v^-$};
\node[va, right=of A+, yshift=-0.5cm] (B){$v$};
\node[left=1cm of A, yshift=-0.5cm](d1){};
\node[left=1cm of A, yshift=0.5cm](d2){};
\node[right=1cm of B, yshift=-0.5cm](d3){};
\node[right=1cm of B, yshift=0.5cm](d4){};
\path
(A) edge node[el]{$p$} (A+)
(A+) edge[loop, looseness=4,out=45,in=135] node[swap,el]{$p+1$} (A+)
(A+) edge node[el]{$p$} (B)
(A) edge node[el,swap]{$p$} (A-)
(A-) edge[loop, looseness=4,out=-45,in=-135] node[el]{$p-1$} (A-)
(A-) edge node[el,swap]{$p$} (B)
(d1) edge[dotted] (A)
(d2) edge[dotted] (A)
(B) edge[dotted] (d3)
(B) edge[dotted] (d4)
;
\end{tikzpicture}
\end{center}
\caption{Vertex gadget for vertex $v \in V \setminus V_\exists$ with even priority $p$}\label{fig:parityGadget}
\end{figure}
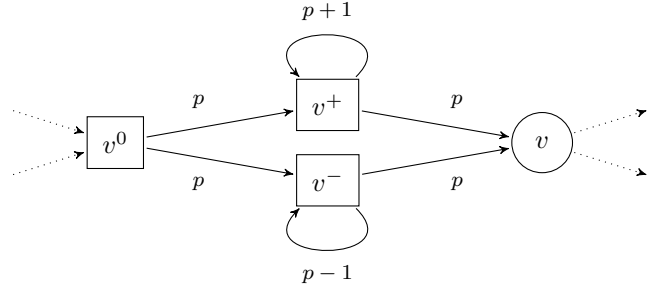

As the weights and interval boundaries are integers in $[0,|V|+1]$ this is clearly a polynomial time translation to a unary-encoded interval mean-payoff game.  We claim that Eve wins the parity game if and only if she wins the interval mean-payoff game.  Suppose she has a positional winning strategy $\sigma$ in the parity game.  We define her strategy $\sigma'$ as follows.  For any vertex $v \in V_\exists$ she moves to the vertex gadget corresponding to the vertex she would have moved to under $\sigma$.  That is, $\sigma'(v) = (\sigma(v),0)$.  Whenever the play reaches a vertex gadget that she controls (i.e.\ a vertex $v^0$ where $v$ has even priority $p$ in the parity game), her strategy is to remain in the gadget until the average weight of the current play lies in the interval $[p,p+\frac{1}{2}]$.  She does this by moving to $v^+$ if the current average is below the interval, and to $v^-$ if the average is above, and then staying at that vertex until the average weight reaches the interval.  Note that after sufficiently many steps this will always be possible.  When the average weight lies in $[p,p+\frac{1}{2}]$ she moves to $v$ and the game continues.  There is a clear $1$-$1$ correspondence between plays consistent with $\sigma$ and plays consistent with $\sigma'$, and if a play in the parity game visits a vertex with even priority $p$ infinitely often, then the running average of the corresponding play will lie in the interval $[p,p+\frac{1}{2}] \subseteq I$ infinitely often.  By construction, Adam can never reduce the mean-payoff below the interval $[p,p+\frac{1}{2}]$ unless the play reaches a gadget corresponding to a vertex of lower priority.  This is important because we use the $\liminf$ definition of mean-payoff.  Further, if he chooses to remain in a gadget indefinitely he will lose.  As all plays consistent with $\sigma$ have the property that the minimal priority visited infinitely often is even, it follows that for all plays $\pi$ consistent with $\sigma'$ there is some even priority $p$ such that $\underline{MP}(\pi) \in [p,p+\frac{1}{2}] \subseteq I$.  Thus the $\sigma'$ is winning for Eve.  For the converse we see that Adam can translate a winning strategy from the parity game in the same manner.

\end{proof}

\subsection{Single interval}\label{sec:singleMP}
We now examine in more detail the case when $I$ is a single interval.   As we can replace any strict threshold call with a non-strict threshold we can assume without loss of generality that $I$ is closed.
%
%
The simplification of Algorithm~\ref{alg:MPfull} to a single closed interval is given in Algorithm~\ref{alg:singleMP}.

\begin{algorithm}

\begin{algorithmic}
\REQUIRE A game graph $G$ and a bounded closed real interval $[a,b]$.
\ENSURE $(W^\exists,W^\forall)$ where $W^\exists$ ($W^\forall$) are the vertices from which Eve (Adam) has a winning strategy.
\STATE $W \leftarrow \emptyset$
\REPEAT
\STATE $(A,A') \leftarrow \xMP_{\geq a}(G)$
\STATE $(B,B') \leftarrow \xMP_{\leq b}(G\setminus A')$
\STATE $W \leftarrow W \cup A' \cup B'$
\STATE $G \leftarrow G \setminus (A' \cup B')$
\UNTIL{$A' \cup B' = \emptyset$}
\RETURN $(V \setminus W, W)$
\end{algorithmic}
\caption{$\xMP_{[a,b]}(G)$}\label{alg:singleMP}
\end{algorithm}



We observe that Algorithm~\ref{alg:singleMP} makes at most a linear number of calls to the mean-payoff threshold problem, so lies in the intersection of $\NP$ and $\coNP$.
\begin{theorem}  
Deciding if Eve wins a single interval mean-payoff game is in $\NP \cap \coNP$.  
\end{theorem}


\subsubsection{Memory considerations}
The strategies for Adam and Eve described in the proof of Lemma~\ref{lem:MPCorrect} require infinite memory.  We now show, with a careful analysis, that in the case of a single interval this can be improved.
\begin{theorem}
Let $(G,I)$ be a single interval mean-payoff game.  If Adam has a winning strategy then he has a positional winning strategy.  If Eve has a winning strategy then she has a strategy that requires finite memory.
\end{theorem}
\begin{proof}
Algorithm~\ref{alg:singleMP} consists of repeatedly removing vertices from which Adam can either ensure the mean-payoff lies above or below $I$.  Clearly Adam has a winning strategy from any vertex removed: he plays his (positional) winning strategy corresponding to the level at which the vertex was removed, until the play reaches a vertex removed at an earlier stage.  We observe that any consistent play will never return to a vertex removed at a later stage (as such vertices are in the winning set for Eve at the same point of the iteration), so this strategy is in fact positional.  Any play consistent with this strategy will eventually stabilize at some stage of the iteration, whereupon Adam's strategy for that stage will ensure the mean-payoff lies outside $I$.  We also observe that this result follows from the fact that the objective is prefix-independent and convex, so from~\cite{Kop07} Adam has a positional winning strategy.

The idea behind Eve's finite memory strategy on $W^\exists$ is to keep track of the total weight seen so far (rather than the average as in the proof of Lemma~\ref{lem:MPCorrect}) \emph{modulo cycles with average weight in $I$}. 
 This ensures, with the strategy outlined below, that the total weight will remain within some bounded range, and hence the strategy will only require finite memory.

By subtracting a constant from the weights of all edges and the interval bounds, we can assume that $0 \in I$. 
We observe on the vertices in $W^\exists$ Eve has two (positional) strategies: $\sigma_<$ which ensures $\underline{MP} \leq \sup I$ and  $\sigma_>$ which ensures $\underline{MP} \geq \inf I$. 
Eve's strategy is to alternate between these two strategies, as in the proof of Lemma~\ref{lem:MPCorrect}, however now she changes when the following condition is met.
We keep a stack-based history of the current play and when a cycle $\chi$ is completed we remove it from the history of the current play, keeping the first vertex of the cycle on the top of the stack.  If $w(\chi)/|\chi| \in I$ we say $\chi$ is \emph{good} and she continues to play her current strategy.  If $w(\chi)/|\chi| \notin I$, she adds $w(\chi)$ to a counter.  Note that if she was playing $\sigma_<$ she would only subtract from the counter and if she was playing $\sigma_>$ then she would only add to the counter because $\sigma_<$ and $\sigma_>$ are winning positional strategies.  She switches strategies if the counter changes sign.  That is, if she was playing $\sigma_{<}$ and the counter value falls below $0$ she switches to $\sigma_>$, and she switches to $\sigma_<$ if she was playing $\sigma_>$ and the counter value goes above $0$.  
Clearly this strategy requires only exponential memory: Eve needs only to store at most $|V|$ vertices in the history and because $\sigma_>$ and $\sigma_<$ are positional the counter values are bounded by $\pm |V|\cdot W$.  We claim that any play $\pi$ consistent with this strategy has $\underline{MP}(\pi) \in I$.  

Let $\pi$ be a play consistent with the strategy.  Let us consider the state of the strategy after $k$ steps of the play.  Let $w_k$ be the total weight of all good cycles popped, and $l_k\leq k$ their total length.  Let $c_k$ denote the counter value.  We observe that the stack contents being stored are always a finite prefix of $\pi$ (when read from bottom to top), so we can define $s_k$, the weight of the stack, as the weight of the corresponding prefix.  It is clear from the definition of the strategy that:
\[ w(\pi[..k]) = w_k + c_k + s_k.\]
Also, $-|V|\cdot W \leq c_k,s_k \leq |V|\cdot W$, and $\frac{w_k}{l_k} \in I$.  As $0 \in I$ we have $\inf I \leq 0 \leq \sup I$, so 
\[ \inf I \leq \frac{l_k (\inf I)}{k} \leq \frac{w_k}{k} \leq \frac{w_k}{l_k} \leq \sup I.\]  Therefore,
\begin{eqnarray*}
\frac{w(\pi[..k])}{k} &\geq& \frac{-2|V|\cdot W}{k} + \inf I\quad\rightarrow\quad\inf I\text{ as $k \rightarrow \infty$, and}\\
\frac{w(\pi[..k])}{k} &\leq& \frac{2|V|\cdot W}{k} + \sup I\quad\rightarrow\quad\sup I\text{ as $k \rightarrow \infty$.}
\end{eqnarray*}
Hence, as $I$ is closed, $\underline{MP}(\pi) \in I$ as required.
\end{proof}

\section{Discount sum games}\label{sec:DS}
In this section we consider interval discount sum games.  Here we make a distinction between whether or not singleton intervals (and singleton gaps between intervals) are permitted, because unlike other payoff functions considered in this paper there is a marked difference between the corresponding games.  We show that for non-singleton intervals the problem of determining the winner is $\PSPACE$-complete and as a consequence of our algorithm we show that finite memory stategies suffice. For singleton intervals (including the exact value problem) our $\PSPACE$-hardness result holds, but is not even known if determining the winner is decidable.  We give a simple example that shows that infinite memory is required for winning strategies in this case. 
\subsection{Single, non-singleton intervals}
We show that the problem for discount sum games in this case is $\PSPACE$-complete for any discount factor $\lambda$.

\paragraph*{Lower bound.}
To show $\PSPACE$-hardness we reduce from the subset sum game defined in~\cite{FJ13}.  The subset sum game is specified by a target $t \in \mathbb{N}$
and a list of pairs of natural numbers $(a_1,a_1'),(a_2,a_2'), \ldots, (a_n,a_n')$.  The game takes $n$ rounds, in round $i$, one player (Adam if $i$ is odd, Eve if $i$ is even) chooses $a_i$ or $a_i'$.  After $n$ rounds Eve wins if and only if the sum of the selected numbers is $t$.  Given an instance of the subset sum game we construct the following interval discount sum game (for discount factor $\lambda$):
\begin{itemize}
\item $V = \{v_1, v_2, \ldots, v_{n+1}\}$, 
\item $V_\exists = \{v_i \st i\text{ is even}\}$,
\item $q_0 = v_1$, 
\item $E$ and $w$ defined as follows:
\begin{itemize}
\item For $1 \leq i \leq n$ there are two edges from $v_i$ to $v_{i+1}$, one with weight $\frac{a_i}{\lambda^{i-1}}$ and one with weight $\frac{a_i'}{\lambda^{i-1}}$,
\item There is a loop with weight $0$ on $v_{n+1}$.
\end{itemize}
\item $I = (t-1,t+1)$
\end{itemize}
The reduction is illustrated in Figure~\ref{fig:subDS}.

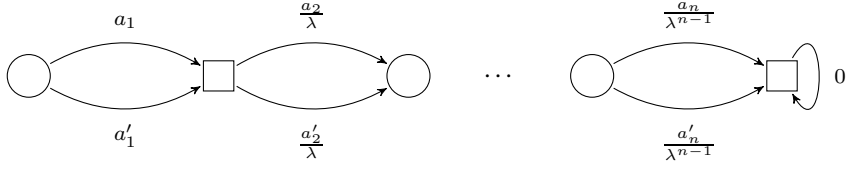
\begin{figure}
\begin{center}
\begin{tikzpicture}[inner sep=2mm, ve/.style={rectangle, draw}, va/.style={circle, draw}]
\node[va] (A){};
\node[ve, right=of A] (B){};
\node[va, right=of B] (C){};
\node[right=of C,xshift=-1.5cm] (X){\ldots};
\node[va,right=of X,xshift=-1.5cm] (D) {};
\node[ve,right=of D] (E) {};

\path
 (A) edge[bend left] node[el]{$a_1$} (B)
(A) edge[bend right] node[el,swap]{$a_1'$} (B)
 (B) edge[bend left] node[el]{$\frac{a_2}{\lambda}$}(C)
(B) edge[bend right] node[el,swap]{$\frac{a_2'}{\lambda}$} (C)
 (D) edge[bend left] node[el]{$\frac{a_n}{\lambda^{n-1}}$} (E)
(D) edge[bend right] node[el,swap]{$\frac{a_n'}{\lambda^{n-1}}$}(E)
(E) edge[loop, looseness=6, out=60, in=-60] node[el]{$0$} (E);

\end{tikzpicture}
\end{center}
\caption{Reduction from subset sum games to interval discount sum games}\label{fig:subDS}
\end{figure}

Note that as $\log\left (\frac{a}{\lambda^n}\right) = \log(a) - n\cdot \log(\lambda)$ the binary representations of the weights on this graph are still polynomial in the size of the input, so this is a polynomial time translation. 
It is clear that a play in this game corresponds to a selection of elements from the pairs, and the discounted sum of the play is equal to the sum of the corresponding elements.  As this sum is always an integer, the discounted sum lies in the interval $(t-1,t+1)$ if and only if the sum is equal to $t$.  Thus this is a polynomial time reduction from subset sum games to interval discounted sum games.

A corollary of this construction is that positional strategies are not sufficient for interval discount sum games.

\paragraph*{Upper bound.}
Given $v \in V$ and strategies $\sigma$ and $\tau$ for Eve and Adam respectively, we define $\val{v}{\sigma}{\tau}$ to be the payoff of the unique play from $v$ consistent with $\sigma$ and $\tau$.
Two important (memoryless) strategies for Eve are $\sigma_{\max}$ and $\sigma_{\min}$, the strategies which, for all states $v$, maximize $\min_{\tau} \val{v}{\sigma}{\tau}$ and minimize $\max_{\tau} \val{v}{\sigma}{\tau}$ respectively.

The idea behind the upper bound centres around the observation that after many steps the remainder of any play does not contribute much to the overall discounted sum.  If the target interval is non-singleton then after sufficiently many steps the problem reduces to the classical threshold problem.  Thus we can stop the game after finitely many steps when it becomes a trivial matter to determine if the overall discounted sum will lie in the interval or not.  The key lemma for the result is the following:
\begin{lemma}\label{lem:finite}
Suppose Eve has a winning strategy to ensure the discounted sum lies in an interval $I$, and let 
\[N =  \left \lfloor \frac{\log(|I|) + \log(1-\lambda) - \log(2W)}{\log \lambda} \right \rfloor\] where $W$ is the maximum absolute value of any weight occurring in $G$. 
Then Eve has a winning strategy that agrees with either $\sigma_{\max}$ or $\sigma_{\min}$ after $N$ steps.
\end{lemma}

Note that whether the strategy agrees with $\sigma_{\max}$ or $\sigma_{\min}$ depends on the play up to the $N$-th step.  It is feasible that against one strategy of Adam this strategy will agree with $\sigma_{\max}$ but against another strategy it will agree with $\sigma_{\min}$.

\begin{proof}
We first observe that $N$ is chosen such that for all $n>N$ we have
\begin{equation}\label{eqn:small}
|I| > \lambda^n \cdot \left( \frac{2W}{1-\lambda} \right).
\end{equation}
That is, after the $N$-th step of any play, the overall contribution of the remainder of the play is restricted to an interval smaller than $I$.

Let $\sigma$ be a winning strategy for Eve.  The desired winning strategy will follow $\sigma$ for $N$ steps and then one of $\sigma_{\max}$ or $\sigma_{\min}$ depending on the value of the play in a manner described presently.  Suppose after $N$ steps the current play has value ${\tt x}$ and is in state $v$.   As $\sigma$ is a winning strategy, we have for any strategy $\tau$ for Adam:
\begin{equation}\label{eqn:strat} 
{\tt x} + \lambda^{N+1} \cdot  \val{v}{\sigma}{\tau} \in I.
\end{equation} 
Now, as $|\val{v}{\sigma}{\tau}| \leq \frac{W}{1-\lambda}$, it follows from (\ref{eqn:small}) and (\ref{eqn:strat}) that at least one of the following is true:
\begin{subequations}\begin{align}
{\tt x} + \lambda^{N+1} \cdot \frac{W}{1-\lambda} &\in I,\text{ or}\label{eqn:below}\\
{\tt x} - \lambda^{N+1} \cdot \frac{W}{1-\lambda} &\in  I.\label{eqn:above}
\end{align}
\end{subequations}
If (\ref{eqn:below}) holds then we follow $\sigma_{\max}$, otherwise we follow $\sigma_{\min}$.  To show that the resulting strategy is winning,
let us suppose (\ref{eqn:below}) holds, the case for (\ref{eqn:above}) being similar.  From the definition of $\sigma_{\max}$ we have, for any state $w$ and any strategy $\tau$ of Adam:
\[ \val{w}{\sigma}{\tau} \leq \val{w}{\sigma_{\max}}{\tau} \leq \frac{W}{1-\lambda}.\]
Hence it follows from (\ref{eqn:strat}) that for any strategy $\tau$ of Adam:
\[{\tt x} + \lambda^{N+1} \cdot \val{v}{\sigma_{\max}}{\tau} \geq {\tt x} + \lambda^{N+1} \cdot  \val{v}{\sigma}{\tau} \in I,\]
and from (\ref{eqn:below}):
\[ {\tt x} + \lambda^{N+1} \cdot \val{v}{\sigma_{\max}}{\tau} \leq {\tt x} + \lambda^{N+1} \cdot \frac{W}{1-\lambda}  \in I.\]
Thus the payoff of any play consistent with this strategy lies in $I$ and is therefore winning for Eve.
\end{proof}

\begin{corollary}
Finite memory strategies are sufficient in non-singleton interval discount sum games.
\end{corollary}

The algorithm for determining the winner of a non-singleton interval discount sum game is straightforward. We run an alternating Turing Machine for $N$ steps to guess an initial
play.  Note that $N$ is polynomial in the size of the input, so this can be done in $\PSPACE$.  Suppose the play ends in state $v$ with the current discounted sum ${\tt x}$.  We compute the four values:
\[ \begin{array}{lcl}
\maxmax = \max_{\tau} \val{v}{\sigma_{\max}}{\tau}&\quad& \minmax =  \min_{\tau} \val{v}{\sigma_{\max}}{\tau} \\ \maxmin =  \max_{\tau} \val{v}{\sigma_{\min}}{\tau} &\quad& \minmin = \min_{\tau} \val{v}{\sigma_{\min}}{\tau}.\end{array}\]
These are computable in $\NP \cap \coNP$: $\minmax$ and $\maxmin$ using the standard algorithm for discount sum games, and $\maxmax$ ($\minmin$) by fixing $\sigma_{\max}$ ($\sigma_{\min}$ respectively), computed in the previous step, and treating the resulting game as a solitaire discount sum game with Adam trying to maximize (minimize) the payoff.  Finally we check if either:
\begin{eqnarray*}
{\tt x} + \lambda^{N+1}\cdot\minmax \in I &\text{ and }& {\tt x} + \lambda^{N+1}\cdot\maxmax \in I,\text{ or}\\
{\tt x} + \lambda^{N+1}\cdot\minmin \in I &\text{ and }& {\tt x} + \lambda^{N+1}\cdot\maxmin \in I.
\end{eqnarray*}
It is clear that one of the above conditions holds if and only if $\sigma_{\max}$ or $\sigma_{\min}$ is winning from the current position.  Therefore, from Lemma~\ref{lem:finite}, one of the above conditions holds if and only if Eve has a winning strategy.

\begin{theorem}
Let $G$ be a game graph, $I \subseteq \mathbb{R}$ a non-singleton real interval and $\lambda \in (0,1)$.  
Deciding if Eve wins the interval discount sum game $(G,I,\lambda)$ is $\PSPACE$-complete.
\end{theorem}

We observe that if the weights, interval bounds and discount factor are all encoded in unary then $N$ is logarithmic in the size of the input and $\maxmax$, $\minmax$, $\maxmin$, and $\minmin$ can all be computed in polynomial time using a pseudo-polynomial time algorithm for the threshold problem for discount sum games (see e.g.~\cite{zp96}).  Thus the above algorithm runs in polynomial time.

\begin{theorem}
Let $G$ be a game graph, $I \subseteq \mathbb{R}$ a non-singleton real interval and $\lambda \in (0,1)$ all encoded in unary.  
Deciding if Eve wins the interval discount sum game $(G,I,\lambda)$ is in $\P\TIME$.
\end{theorem}

\subsection{Multiple intervals}
The algorithm of the previous section also applies to multiple intervals \emph{as long as the gaps between the intervals are also non-singleton}.  This follows from the observation that after sufficiently many steps the overall discount payoff will not deviate too far from the current value, so at that point the game reduces to the single interval case. 
\begin{theorem}
Let $G$ be a game graph, $I$ a finite union of real intervals such that neither $I$ nor $\mathbb{R}\setminus I$ contains singleton elements, and $\lambda \in (0,1)$.  
Deciding if Eve wins the interval discount sum game $(G,I,\lambda)$ is $\PSPACE$-complete.
\end{theorem}

\subsection{Singleton intervals}
When the set of intervals (or their complement) include singleton intervals, the situation is more complicated.  Following the same argument as the previous section, after sufficiently many steps the problem reduces to the exact value problem: Given a game graph $G$, a discount factor $\lambda \in \mathbb{Q}$ and a target $t \in \mathbb{Q}$, does Eve have a strategy to ensure the discounted sum is exactly $t$?

It is currently open whether this problem is even decidable, however the \PSPACE-hardness result from the previous section (using the interval $\{t\}$ rather than $(t-1,t+1)$) gives a lower-bound.  The problem is related to the universality problem for discount sum automata~\cite{bo14}, a well-known problem for which decidability remains open~\cite{bh11}.  The problem was also studied for Markov Decision Processes and graphs (i.e.\ one-player games) in~\cite{cfw13} where it was shown to be decidable for discount factors of the form $\lambda = \frac{1}{n}$, and that in general infinite memory is required.

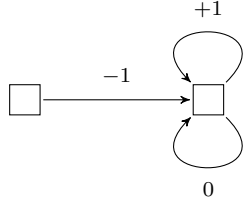
\begin{figure}
\begin{center}
\begin{tikzpicture}[inner sep=2mm, ve/.style={rectangle, draw}, va/.style={circle, draw}]
\node[ve] (A){};
\node[ve, right=of A] (B){};
\path
(A) edge node[el]{$-1$} (B)
(B) edge[loop, looseness=8,in=135,out=45] node[swap,el]{$+1$} (B)
(B) edge[loop, looseness=8,in=-135,out=-45] node[el]{$0$} (B);
\end{tikzpicture}
\caption{Exact value discount sum game ($\lambda = \frac{2}{3}$, $t=0$) that requires infinite memory (modified from~\cite{cfw13})}\label{fig:exactInfMem}
\end{center}

\end{figure}

\begin{lemma}[\cite{cfw13}]
There exist exact value discount sum games for which an infinite memory is required for a winning strategy.
\end{lemma}
\section{Total sum games}\label{sec:totalSum}
Total sum games refine mean-payoff games and can be seen as a special case of discount sum games where the discount factor is $1$.  
Assuming the graph has integer weights, $\underline{Total}$ will always be an integer (or $\pm \infty$), thus we can assume all intervals are closed or open as necessary.

The objective of total sum games is similar to reachability in one-dimensional vector addition systems with states~\cite{bjk10} and counter reachability games~\cite{Rei13}, however we are interested in values seen infinitely often rather than reaching a particular state and counter value.  The complexity bounds we obtain are similar to these problems, indeed we use the same problems for establishing the bounds.  However it is not clear if there is a more direct reduction between these problems.


\subsection{Lower bounds}
In this section we establish the following result:
\begin{theorem}
\begin{itemize}
\item The problem of deciding if Eve wins an interval total sum game is $\EXP$-hard.  
\item The problem of deciding if Eve wins a unary-encoded interval total sum game is $\PSPACE$-hard.
\end{itemize}
\end{theorem} 

\subsubsection{Binary encoding}
We first show that deciding the winner of interval total sum games is $\EXP$-hard by reducing from \emph{countdown games}.  A countdown game is played on a weighted graph, where all weights are negative.  The play starts by setting a counter to a given initial value.  Whenever an edge is taken the counter is decremented by the weight.  Eve wins if and only if she reaches a vertex with the counter exactly $0$.  Deciding the winner of countdown games is known to be $\EXP$-complete~\cite{jsl08}.  Note that by subdividing edges if necessary we can assume that the players play alternately, that is the graph is bipartite.  The reduction is straightforward, given a countdown game $G$ with initial credit $c$ we construct the following total sum game.  We add two new vertices (of Eve) $v_I$ and $v_\bot$.  There is an edge from $v_I$ to the initial vertex of $G$ with weight $c$, and an edge of weight $0$ from every vertex of Eve in $G$ to $v_\bot$.  Also, for every edge $e=(v,v')$ where $v$ is a vertex of Eve we add another edge $(v,v_\bot)$ of weight $w(e)$.  Finally we have an edge $(v_\bot,v_\bot)$ of weight $0$.  Clearly Eve can ensure $\underline{Total} = 0$ if and only if she can reach $v_\bot$ with a total sum of $0$.  Thus she can win the interval total sum game, with interval $\{0\}$, if and only if she can win the countdown game.

\subsubsection{Unary encoding}
For unary-encoded interval total sum games, we reduce from the non-emptiness problem for one letter alphabet alternating automata, shown to be \PSPACE-complete in~\cite{Hol95}.  Again, the reduction is simple as this problem can be viewed as a countdown game where all edges have weight $-1$ and Eve has to guess the initial credit.  The guessing stage can be implemented by having a loop on $v_I$ with weight $+1$.  The remainder of the reduction is as in the reduction from countdown games.
\subsection{Upper bound}
We now show that interval total sum games can be solved in $\EXPSPACE$ by reducing them to parity games on infinite graphs described by one-counter machines.  
Such games were studied in~\cite{Ser06} where determining the winner was shown to be decidable in $\PSPACE$, but the graphs were described by a \emph{unary} counter machine, or equivalently, pushdown graphs with a single-letter alphabet.
Here we use a definition corresponding to the use of a binary-valued counter (also called long-range in~\cite{Rei13}).  More formally, a one-counter game graph is described by a tuple $(V,V_\exists,E,E_0,w,q_0)$ where $(V,E,w)$ is a finite weighted graph, $V_\exists \subseteq V$, $E_0 \subseteq V \times V$ and $q_0 \in V$.  The (infinite) unweighted game graph corresponding to such a tuple is $(V \times \mathbb{Z}, V_\exists \times \mathbb{Z}, E', (q_0,0))$ where $E'$ is defined as follows:
\begin{itemize}
	\item If $e = (v,v') \in E$ then for all $c \in \mathbb{Z}$, $\big((v,c),(v',c+w(e))\big) \in E'$, and
	\item If $(v,v') \in E_0$ then $\big((v,0),(v',0)\big) \in E'$.
\end{itemize} 
Intuitively a one-counter game graph is a game graph augmented with a counter which is incremented or decremented by weights on traversed edges.  A special set of edges, $E_0$, are activated only if the counter has value $0$.
It is clear a binary one-counter graph can be described by an exponentially larger unary one-counter graph\footnote{We allow negative counter values, but this can be handled with non-negative counter values by doubling the state space}, hence our reduction yields an $\EXPSPACE$ algorithm.

The key observation for the reduction is that interval total sum games can be viewed as parity games on $V \times \mathbb{Z}$, where the second component keeps track of the total sum seen so far.  The priority of a vertex $(v,c)$ is determined by which interval (or gap between intervals) contains $c$, in the same manner used in the equivalence between liminf games and parity games in Section~\ref{sec:liminf}.  However, we cannot use the result on parity games on one-counter graphs directly for this observation because for those games the priorities are defined by the states of the counter-machine and not the values of the counter.  Instead, we have Eve assert which interval (or gap between intervals) the counter is in, and give Adam the ability to punish her if she claims falsely.

Let $(V,V_\exists,E,w,q_0,I)$ be an interval total sum game.  Recall from Section~\ref{sec:liminf} the definition of $\Omega_I$.  Let us define $m_i := \min \Omega_I^{-1}(i)$ and $M_i := \max \Omega_I^{-1}(i)$.  We construct a parity game on a one-counter graph $(V',V'_\exists,E',E'_0,w',q_0',\Omega)$ as follows.
\begin{itemize}
\item $V' = (V \times \{0,1\}) \times [1,2r+1] \cup \{v_e \st e \in E\} \cup \{v_0,v_\bot,v_\top\}$;
\item $V'_\exists = E \cup \{v_0,v_\bot,v_\top\} \cup \{(v,1,i) \st v \in V_\exists\text{ and }i \in [1,2r+1]\}$;
\item $q_0' = (q_0,1,\Omega_I(0))$;
\item $E_0' = \{(v_\bot,v_0), (v_\top,v_0)\}$;
\item $E'$ and $w'$ given as follows, for all $i \in [1,2r+1]$:
\begin{itemize}
\item For every $e = (v,v') \in E$, an edge from $(v,1,i)$ to $v_e$ with weight $w(e)$ and an edge from $v_e$ to $(v',0,i)$ with weight $0$,
\item An edge from $(v,0,i)$ to $v_\bot$ with weight $-m_i$ if $m_i > -\infty$,
\item An edge from $(v,0,i)$ to $v_\top$ with weight $-M_i$ if $M_i < \infty$,
\item An edge from $(v,0,i)$ to $(v,1,i)$ with weight $0$, and
\item Loops on $v_\bot$, $v_\top$ and $v_0$ with weights $-1$, $+1$ and $0$ respectively.
\end{itemize}
\item $\Omega((v,0,i)) = \Omega((v,1,i)) = \Omega_I(i)$, $\Omega(v_e) = \Omega(v_\bot) = \Omega(v_\top) = 2r+1$, and $\Omega(v_0) = 2r$.
\end{itemize}	
Intuitively, we create $2r+1$ copies of the game graph (one for each interval and one for each gap), but replace edges with the edge gadget shown in Figure~\ref{fig:parityCounter}.
\begin{figure}
\begin{center}
\begin{tikzpicture}[inner sep=2mm, ve/.style={rectangle, draw}, va/.style={circle, draw}, node distance=1cm]
\node[ellipse,draw](A){$v,1,i$};
\node[ve,right=1.5cm of A](e){$v_e$};
\node[ellipse,draw,right=1.5cm of e](b){$v',0,i'$};
\node[ve,right=of b](B){$v',1,i'$};
\node[ve,below=0.7cm of b, xshift=-1.5cm](bot){$v_\bot$};
\node[ve,below=0.7cm of b, xshift=1.5cm](top){$v_\top$};
\node[ve,below=0.7cm of bot, xshift=1.5cm](z){$v_0$};

\node[left=0.5cm of e,yshift=0.5cm](d1){};
\node[left=0.5cm of e,yshift=-0.5cm](d2){}; 
\node[right=0.5cm of e,yshift=-0.5cm](d3){};
\node[right=0.5cm of e,yshift=0.5cm](d4){}; 

\path
(A) edge node[el,pos=0.3]{$w(e)$} (e) 
(e) edge node[el,pos=0.7]{$0$} (b) 
(b) edge node[el]{$0$} (B)

(b) edge node[el,swap,anchor=east]{$-m_{i'}$} (bot)
(b) edge node[el,anchor=west]{$-M_{i'}$} (top)
(bot) edge[dashed] node[el,swap,anchor=east]{$=0?$} (z)
(top) edge[dashed] node[el,anchor=west]{$=0?$} (z)

(d2) edge[dotted] (e)
(e) edge[dotted] (d3)
(d1) edge[dotted] (e)
(e) edge[dotted] (d4)

(bot) edge[loop,looseness=4,out=-135,in=135] node[el]{$-1$} (bot)
(top) edge[loop,looseness=4,out=-45,in=45] node[el,swap]{$+1$} (top)
(z) edge[loop,looseness=4,out=-45,in=-135] node[el]{$0$} (z)
;
\end{tikzpicture}
\caption{Edge gadget for edge $e = (v,v')$,  $v \notin V_\exists$, $v' \in V_\exists$}\label{fig:parityCounter}
\end{center}
\end{figure}
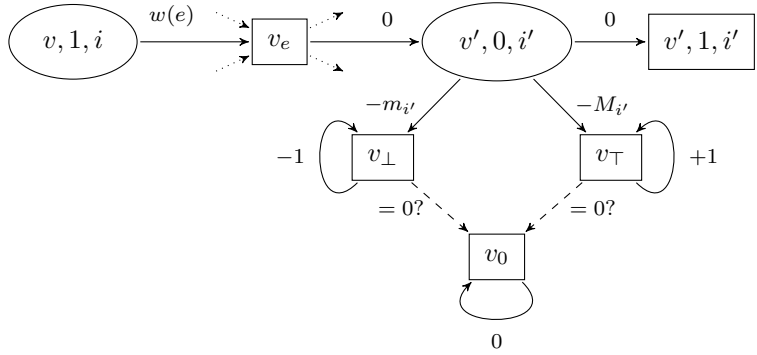

We now show that Eve has a winning strategy in this parity game if and only if she has a winning strategy in the interval total sum game.  We first observe that if $v_\bot$ ($v_\top$) is reached with a negative (positive) counter value then the edge to $v_0$ is never activated so the vertex acts as a sink which is winning for Adam.  Conversely, if $v_\bot$ ($v_\top$) is reached with a non-negative (non-positive) counter value then the loop decrements (increments) the counter until the edge to $v_0$ is activated, whereupon Eve can win by moving to this sink which is winning for her.  It follows that if the play reaches a vertex $(v,0,i)$ and the counter value is outside $[m_i,M_i]$ then Adam can win by playing to $v_\bot$ if the counter is $<m_i$ or to $v_\top$ if the counter is $>M_i$.  On the other hand, if the counter is in the range $[m_i,M_i]$ then Eve wins if Adam plays to either of these vertices.  Thus the gadget defined by the vertices $\{v_\bot,v_\top,v_0\}$ allows Adam to punish Eve if the counter is not in the asserted interval and lets Eve win if Adam attempts to falsely punish her.  Now, assuming Eve plays correctly, it is easy to see that the minimal priority seen infinitely often corresponds to the lowest interval or interval gap visited infinitely often by the counter. Thus Eve has a winning strategy in the parity game if and only if she has a winning strategy in the interval game.

\begin{theorem}
Deciding if Eve wins an interval total sum game is in \EXPSPACE.
\end{theorem}

We conclude by observing that if the interval game is encoded in unary, then the above reduction is a polynomial time reduction to the parity games on one-counter graphs considered in~\cite{Ser06}, giving an upper bound to match our lower bound.
\begin{theorem}
Deciding if Eve wins a unary encoded interval total sum game is \PSPACE-complete.
\end{theorem}

\subsection{Memory requirements}
We now consider memory requirements for winning strategies in interval total sum games.  We show that, in general, infinite memory is required for winning strategies, but for single interval games, winning strategies for Eve need only finite memory.

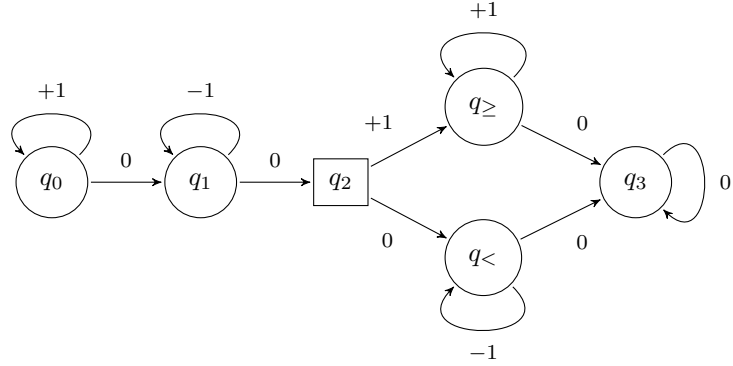
\begin{figure}
\begin{center}
\begin{tikzpicture}[inner sep=2mm, ve/.style={rectangle, draw}, va/.style={circle, draw}, node distance=1cm]
\node[va] (0){$q_0$};
\node[va, right=of 0] (A){$q_1$};
\node[ve, right=of A] (B){$q_2$};
\node[va, right=of B, yshift=1cm] (C){$q_{\geq}$};
\node[va, right=of B, yshift=-1cm] (D){$q_<$};
\node[va, right=of C, yshift=-1cm] (E){$q_3$};

\path
(0) edge[loop,looseness=4, out=45, in=135] node[el,swap]{$+1$} (0)
(0) edge node[el]{$0$} (A)
(A) edge[loop,looseness=4, out=45, in=135] node[el,swap]{$-1$} (A)
(A) edge node[el]{$0$} (B)
(B) edge node[el]{$+1$} (C)
(B) edge node[el,swap]{$0$} (D)
(C) edge[loop,looseness=4, out=45, in=135] node[el,swap]{$+1$} (C)
(D) edge[loop,looseness=4, out=-45, in=-135] node[el]{$-1$} (D)
(C) edge node[el]{$0$} (E)
(D) edge node[el,swap]{$0$} (E)
(E) edge[loop,looseness=4, out=45, in=-45] node[el]{$0$} (E)
;
\end{tikzpicture}

\end{center}
\caption{Interval total sum game ($I = \mathbb{R} \setminus \{0\}$) which requires infinite memory}\label{fig:TSInf}
\end{figure}

\begin{lemma}
Finite memory winning strategies are not sufficient in interval total sum games.
\end{lemma}
\begin{proof}
Consider the game in Figure~\ref{fig:TSInf} with the intervals $(-\infty,0) \cup (0, \infty)$.  Eve has a winning strategy in this game: if the play ever reaches $q_2$ then she moves to $q_{\geq}$ if the total sum is non-negative, and moves to $q_{<}$ otherwise.  Clearly any play consistent with this strategy will have $\underline{Total} \neq 0$ so it is winning for Eve.  Now suppose Eve plays a finite memory strategy.  It follows there exists a memory state which cannot distinguish between two distinct sums at $q_0$.  Therefore, after taking sufficiently many loops at $q_1$, it follows that there exists a memory state which cannot distinguish between two sums of different signs at $q_2$.  As Eve's play depends only on her location and memory state, there is a total value for which Eve makes the ``wrong choice'', i.e.\ she either moves to $q_{\geq}$ with a negative sum or to $q_<$ with a non-negative sum.  Adam's winning strategy is then to play to this move of Eve and then to increase or decrease the total sum to $0$ before moving to $q_3$.
\end{proof}

By exchanging the roles of the players and complementing the interval, we see that even for single interval games Adam may require infinite memory.  We now show this is not the case for Eve.  In fact, we show that having unbounded intervals is necessary for Eve to not have a finite memory winning strategy.

\begin{lemma}
Let $(G,I)$ be an interval total sum game where $I \cap \mathbb{Z}$ is finite.  If Eve has a winning strategy then she has a finite memory winning strategy.
\end{lemma}
\begin{proof}
As observed in the previous section, we can regard an interval total sum game as a parity game on $V \times \mathbb{Z}$.  It is well known~\cite{Zie98} that positional strategies suffice in parity games, even on infinite graphs.  However, in our case such a strategy would depend on the current state \emph{and on the counter value}, so it would not immediately be realizable with finite memory.  We now show that if $I \cap \mathbb{Z}$ is finite and Eve has a winning strategy then we only need to consider a bounded set of counter values so we can realize the strategy with finite memory.  Let $\sigma$ be a positional winning strategy for Eve on $V \times \mathbb{Z}$, and let $\overline{I} = [\inf I, \sup I]$.  If $\sigma$ is winning from $(v,c)$ where $c \notin \overline{I}$ we claim she only requires finite memory to reach a state $(v',c')$ from which $\sigma$ is winning and where $c' \in \overline{I}$.  Consider the finitely-branching, infinite tree of plays consistent with $\sigma$ from $(v,c)$\footnote{That is, the tree rooted at $(v,c)$ where the branches are all the plays consistent with $\sigma$ and a branching occurs when Adam has a choice of moves}.  Let us cut a branch when it first reaches a vertex $(v',c')$ with $c' \in \overline{I}$.  Note that as we are following plays consistent with $\sigma$, such a state is in the winning set of $\sigma$.  We claim the resulting tree is finite.  If it were not, then by K\"onig's lemma there exists an infinite branch, that is, an infinite play consistent with $\sigma$ that does not reach a vertex $(v',c')$ with $c' \in \overline{I}$.  As all even priority states are only of the form $(v',c')$ where $c \in I$, such a play is winning for Adam, contradicting the fact that $\sigma$ is a winning strategy for Eve.  This finite tree then serves as the memory states for the strategy to reach $\overline{I}$ from $(v,c)$.  The finite memory strategy is now clear: if the current state is $(v,c)$ with $c \in \overline{I}$ she moves to $\sigma(v,c)$.  If the play ever reaches a state $(v',c')$ with $c' \notin \overline{I}$ she plays her finite memory strategy until the play returns to $(v'',c'')$ with $c'' \in \overline{I}$.  As $\sigma$ is positional, there are at most $|V| \times |\overline{I}|$ of these ``out-of-bounds'' states reachable (and possibly the initial state $(q_0,0)$) so overall we only require finite memory.
\end{proof}

To complete the argument for single interval total sum games, we observe that if the interval is infinite then we are considering the classical threshold problem for total sum games.  Positional strategies for these games were shown to be sufficient in~\cite{gz04}.

\begin{theorem}
Let $(G,I)$ be a single interval total sum game.  If Eve has a winning strategy then she has a finite memory winning strategy.
\end{theorem}

\end{document}